\documentclass{siamltex}

\usepackage{amsfonts, amsmath, amssymb, psfrag, graphicx, bbm, mathrsfs}

\usepackage[colorlinks=true]{hyperref}
\hypersetup{urlcolor=blue, citecolor=blue, linkcolor=blue}

\newtheorem{remark}[theorem]{Remark}

\newtheorem{assumption}{Assumption}
\newtheorem{result}{Theorem}

\def\fin{\ifmmode{\Large$\diamond$}\else{\unskip\nobreak\hfil
    \penalty50\hskip1em\null\nobreak\hfil{\Large$\diamond$}
    \parfillskip=0pt\finalhyphendemerits=0\endgraf}\fi}

\def\be#1#2\ee{\begin{equation}\label{eq:#1}#2\end{equation}}
\def\req#1{{\rm(\ref{eq:#1})}}
\def\bdm  {\begin{displaymath}}
  \def\edm  {\end{displaymath}}
\def\bdmal{\begin{displaymath}\begin{aligned}}
    \def\edmal{\end{aligned}\end{displaymath}}

\mathchardef\PhiG="0108

\renewcommand{\L}{{\mathscr L}}
\newcommand{\N}{{\mathord{\mathbb N}}}
\newcommand{\R}{{\mathord{\mathbb R}}}

\renewcommand{\L}{{\mathscr L}}

\newcommand{\norm}[1]{\|#1\|}

\newcommand{\rmd}{\,\mathrm{d}}

\newcommand{\eps}{\varepsilon}

\def\req#1{{\rm(\ref{eq:#1})}}

\makeatletter
\newcommand{\dupdots}{\mathinner{\mkern1mu\raise\p@
    \vbox{\kern7\p@\hbox{.}}\mkern2mu
    \raise4\p@\hbox{.}\mkern2mu\raise7\p@\hbox{.}\mkern1mu}}
\makeatother
\makeatletter
\newcount\c@MaxMatrixCols \c@MaxMatrixCols=10
\newenvironment{cmatrixc}{%
  \hskip -\arraycolsep\array{*\c@MaxMatrixCols c}%
}{%
  \endarray \hskip -\arraycolsep
}
\makeatother
\newenvironment{cmatrix}{\left[\cmatrixc}{\endmatrix\right]}

 {\endMakeFramed}

\newcommand{\utilde}{{\widetilde u}}
\newcommand{\ftilde}{{\widetilde f}}
\newcommand{\dtilde}{{\widetilde d}}
\newcommand{\ktilde}{{\widetilde k}}
\newcommand{\Atilde}{{\widetilde A}}
\newcommand{\Dtilde}{{\widetilde D}}
\newcommand{\Ktilde}{{\widetilde K}}
\newcommand{\bfrhotilde}{{\widetilde\bfrho}}
\newcommand{\V}{{{\mathscr V}_u}}
\newcommand{\X}{{\mathscr X}}
\newcommand{\Y}{{\mathscr Y}}

\newcommand{\RR}{{\boldsymbol{R}}}

\newcommand{\bfrho}{{\boldsymbol{\rho}}}

\newcommand{\bfalpha}{{\boldsymbol{\alpha}}}
\newcommand{\bfphi}{{\boldsymbol{\varphi}}}

\newcommand{\dr}{\,\rmd r}
\newcommand{\dR}{\,\rmd\!R}
\newcommand{\dRR}{\,\rmd\!\RR}

\newcommand{\uu}{{u_*}}
\newcommand{\uo}{{u^*}}
\newcommand{\rs}{s}
\newcommand{\cbeta}{c_\beta}
\newcommand{\cbetad}{c_{\beta,d}}
\newcommand{\Cbeta}{C_\beta}
\newcommand{\go}{c_{\beta,d}^*}
\newcommand{\bfe}{{\boldsymbol{e_1}}}

\makeatletter
\renewcommand\@biblabel[1]{#1.}
\makeatother

\title{Fr\'echet differentiability of molecular distribution functions I.
       \boldmath{$L^\infty$} analysis}
\author{Martin Hanke\thanks{Institut f\"ur Mathematik, Johannes
    Gutenberg-Universit\"at Mainz, 55099 Mainz, Germany
    ({\tt hanke@math.uni-mainz.de}). The research leading to this work
    has been done within the 
    Collaborative Research Center TRR 146; corresponding funding 
    by the DFG is gratefully acknowledged.}}

\begin{document}
\sloppy
\maketitle

\begin{abstract}
For a grand canonical ensemble of classical point-like particles at equilibrium
in continuous space we investigate the functional relationship between 
a stable and regular pair potential describing the interaction of the 
particles and the corresponding molecular distribution functions. 
For certain admissible perturbations of the pair potential and sufficiently 
small activity we rigorously establish Frechet differentiability with respect 
to the supremum norm in the image space -- both for bounded domains and in the
thermodynamical limit. 
\end{abstract}

\begin{keywords}
  Statistical mechanics, molecular distribution function, Fr\'echet derivative,
  Kirkwood-Salsburg equations
\end{keywords}

\begin{AMS}
  {\sc 82B21, 82B80}
\end{AMS}

\hspace*{-0.7em}
{\footnotesize \textbf{Last modified.} \today}

\pagestyle{myheadings}
\thispagestyle{plain}
\markboth{M. HANKE}
{FRECHET DIFFERENTIABILITY OF MOLECULAR DISTRIBUTION FUNCTIONS}

\addtocounter{footnote}{2}

\section{Introduction}
\label{Sec:Introduction}
We consider the grand canonical description of a continuous system of identical
classical particles in thermodynamical equilibrium; cf., e.g.,
Hansen and McDonald~\cite{HaMcD13}.
It is assumed that the potential
energy of the system is determined by a pair potential which only 
depends on the distance of the interacting particles. Under certain
additional assumptions on the potential (known as \emph{stability} and
\emph{regularity}) it has been rigorously proved in the 60's of the previous
century (see the monograph by Ruelle~\cite{Ruel69})
that the corresponding molecular distribution functions of the 
particles have a well-defined thermodynamical limit for a small enough
activity coefficient.
For example, the limiting singlet distribution function 
determines the (constant) number density of the particles; 
the corresponding pair distribution function which provides the
probability density of observing two particles in prescribed coordinates
at the same time only depends on the distance between the given coordinates
and gives rise to a so-called \emph{radial distribution function}. 
The \emph{inverse problem} whether a given radial distribution
function can be obtained as the thermodynamical limit of an equilibrium 
distribution for a certain pair potential is an open problem; 
there are only few partial results, 
for example by Henderson~\cite{Hend74} on uniqueness
and by Koralov~\cite{Kora07} on existence of corresponding solutions.

This inverse problem is fundamental for the development of efficient
multiscale algorithms for the numerical simulation of complex soft matter
phenomena, compare, for example, R\"uhle et al~\cite{RJLKA09}. Many of
these algorithms employ methods for coarse-graining complex molecules
and need to derive effective potentials for the coarse-grained \emph{beads} 
from measured data such as the radial distribution function. 

One of the algorithms for solving this inverse problem is the so-called
\emph{Inverse Monte-Carlo} method by Lyubartsev and Laaksonen~\cite{LyLa95},
which utilizes the well-known Newton method for the numerical solution of
nonlinear equations. As such, this method requires the 
derivative of the radial distribution function with respect to the pair
potential, and a formal discretized computation of this derivative for a
canonical ensemble is given in \emph{loc.\,cit.} 

Apparently a rigorous justification of this formula is yet lacking. 
This is the purpose of this present work, which is the first part of two
consecutive papers. A crucial ingredient is the choice of a proper
(and natural) topology for suitable perturbations of a given potential 
(see Proposition~\ref{Prop:U} below). The topology that we suggest allows
to exploit the well-known \emph{Kirkwood-Salsburg system} of equations 
and the corresponding theory in \cite{Ruel69} to rigorously determine 
the derivative of all molecular distribution functions (and their 
thermodynamical limit) with respect to the
underlying potential in the $L^\infty$ norm of the corresponding distribution
functions. We refer to Section~\ref{Sec:Problem} for a precise statement
of our results. 

As far as the radial distribution function is concerned 
it is also natural to investigate the differentiability of the so-called Ursell 
function (sometimes also referred to as 
\emph{pair correlation function}~\cite{HaMcD13}). 
In the thermodynamical limit the Ursell function is 
known to belong to $L^1(\R^3)$ as a function of the distance between 
its two input particle positions, cf.~\cite{Ruel69} again.
We can prove that the differentiability of the Ursell function extends to 
this topology, either. Since the proof requires a completely different set 
of tools we postpone this and related results to the follow-up 
paper~\cite{Hank16b}.

The outline of this first part is as follows. In the following section
the setting and basic assumptions of this work will be specified. There we also
formulate the two main results, Theorem~\ref{Thm:A} on the differentiability
of the molecular distribution functions, and Theorem~\ref{Thm:B} on the
thermodynamical limit of these derivatives. 
Section~\ref{Sec:A} is devoted to a proof of
Theorem~\ref{Thm:A} and Section~\ref{Sec:B} provides the proof of 
Theorem~\ref{Thm:B}. In the final Section~\ref{Sec:explicit} we reconsider
the explicit formula from~\cite{LyLa95} for the derivative of the 
pair distribution function in a bounded domain, and also provide a formula 
for the derivative of the singlet distribution function.

\section{Problem setting}
\label{Sec:Problem}
In the sequel we present our basic assumptions on the system under 
consideration and review some basic facts; most of them are well known, 
and unless stated otherwise, we refer to \cite{Ruel69} as a standard reference.

We consider a grand canonical ensemble of identical classical
point-like particles in a box $\Lambda\subset\R^3$ 
in thermodynamical equilibrium with defined positive inverse temperature
$\beta$ and activity $z$. Throughout this work we assume that
the box $\Lambda$ is a cube centered at the origin.
The interaction of the particles is given
by a pair potential $u:\R^+\to\R$, which only depends on the distance
between the corresponding particles. Such a pair potential is called 
\emph{stable}, if there exists a constant $B>0$ such that
\bdm
   \sum_{1\leq i<j\leq N} \!\! u(|R_i-R_j|) \,\geq\, -BN
\edm
for all configurations of $N$ labeled particles (and all $N\in\N$), 
where $\RR_N=(R_1,\dots,R_N)\in(\R^3)^N$ is the $N$-tupel with the
particle coordinates. A stable pair potential $u$ is called \emph{regular}, 
if the associated \emph{Mayer function}
\be{Mayer}
   f(R) \,=\, e^{-\beta u(|R|)} - 1
\ee 
belongs to $L^1(\R^3)$, i.e., if
\bdm
   \int_0^\infty |e^{-\beta u(r)}-1|\,r^2\!\dr \,<\, \infty\,.
\edm
   
In order to investigate differentiability with respect to $u$ we need to
allow some variability of the potential without disturbing the above two
properties. Therefore we will stipulate a slightly more restrictive 
but much more handy assumption on $u$.

\begin{assumption}
\label{Ass:u}
There exists $\rs>0$ and positive decreasing functions $\uu,\uo:\R^+\to\R$ 
with
\begin{align*}
   \int_0^\rs \uu(r)\,r^2\!\dr \,=\, \infty \qquad &\text{and} \qquad 
   \int_\rs^\infty \uo(r)\,r^2\!\dr \,<\, \infty\,,
\intertext{such that $u$ satisfies}
   u(r) \,\geq\, \uu(r)\,, \ \ r \,\leq\, \rs\,,\qquad &\text{and} \qquad
   |u(r)| \,\leq\, \uo(r)\,, \ \ r \,\geq\, \rs\,.
\end{align*}
\end{assumption}

An example of a pair potential that satisfies Assumption~\ref{Ass:u}
is the familiar \emph{Lennard-Jones potential}
\bdm
   u_{\text{LJ}}(r) \,=\, 
   4\eps\Bigl( \bigl(\frac{\sigma}{r}\bigr)^{12}  
               - \bigl(\frac{\sigma}{r}\bigr)^6\Bigr)
\edm
with parameters $\eps,\sigma>0$.

Given that the pair potential $u$ satisfies Assumption~\ref{Ass:u} 
we introduce the space $\V$ of \emph{perturbations} $v:\R^+\to\R$ for
which $|v|/u$ is bounded in $(0,s)$ and $|v|/\uo$ is bounded in $(s,\infty)$;
$\V$ is a Banach space when equipped with the norm
\be{V}
   \norm{v}_{\V} \,=\, 
   \max\{\,\norm{v/u}_{(0,s)},\norm{v/\uo}_{(s,\infty)}\,\}\,.
\ee
Here, and throughout, the notation $\norm{\,\cdot\,}_\Omega$ refers to the
supremum norm of a function acting from some interval $\Omega\subset\R^d$ 
into a given Banach space. As we see next, the topology of $\V$ defines an
open neighborhood of stable and regular pair potentials around the given $u$.

\begin{proposition}
\label{Prop:U}
Let $u$ satisfy Assumption~\ref{Ass:u}, $\V$ be the Banach space with
norm $\req{V}$, and $0<t_0<1$ be given. 
Then 
there are constants $B\geq 0$ and $\cbeta>0$ such that the potential
$\utilde=u+v$ satisfies
\be{B}
   \sum_{1\leq i<j\leq N} \!\! \utilde(|R_i-R_j|) \,\geq\, -BN
\ee
for every $\RR_N\in(\R^3)^N$ and
\be{cbeta}
   4\pi\int_0^\infty |e^{-\beta \utilde(r)}-1|\,r^2\!\dr \,\leq\, \cbeta
\ee
for every $v\in\V$ with $\norm{v}_\V\leq t_0$, i.e.,
$\utilde$ is stable and regular.
Moreover, for all $N\geq 2$ and all $\RR_N\in(\R^3)^N$
there exists $j^*=j^*(\RR_N)$ such that
\be{jstar}
   \sum_{i=1\atop \,\,\,i\neq j^*}^N \utilde(|R_i-R_{j^*}|) \,\geq\, -2B
   \ee
for every $\utilde=u+v$ with $\norm{v}_\V\leq t_0$; 
in particular, $\utilde$ is bounded by $-2B$ from below.
\end{proposition}

\begin{proof}
The crucial observation is, that for a given $t_0\in (0,1)$,
every $0<r\leq s$, and every $v\in\V$ with $\norm{v}_\V\leq t_0$ there holds
\begin{align*}
   \utilde(r) &\,\geq\, u(r)-|v(r)| \geq (1-t_0) u(r) \geq q \uu(r)
\intertext{with $q=1-t_0>0$; at the same time we have}
   |\utilde(r)| &\,\leq\, |u(r)|+|v(r)| \,\leq\, (1+t_0) \uo(r)
   \,\leq\, (2-q) \uo(r)
\end{align*}
for every $r\geq s$. From this we conclude that under the given assumptions
on $u$ and $v$ the regularity condition~\req{cbeta} does hold true
for some constant $\cbeta>0$ and all $v$ with $\norm{v}_\V\leq t_0$.
Concerning the stability of the
pair potential $\utilde$ we refer to the argument utilized by 
Fisher and Ruelle in \cite{FiRu66}: Their construction provides a universal
even minorant $\underline{u}:\R\to\R$ of positive type such that
\bdm
   \utilde(r) \,\geq\, \underline{u}(r)\,, \qquad r>0\,,
\edm
for all $v\in\V$ with $\norm{v}_\V\leq t_0$. Inequalities \req{B} and 
\req{jstar} then follow with $B=\underline{u}(0)/2$.
\end{proof}

Writing
\be{pairpotential}
   U_N(R_1,\dots,R_N) \,=\, \!\!\sum_{1\leq i<j\leq N} u(|R_i-R_j|)
\ee
for the configurational Hamiltonian of the system,
the \emph{molecular distribution function} $\rho_\Lambda^{(m)}$ for $m$
particles, $m\in\N$, being distributed in $\Lambda$ is defined to be
\be{rho-m}
   \rho^{(m)}_\Lambda(\RR_m)
   \,=\, \frac{1}{\Xi_\Lambda}
         \sum_{N=m}^\infty \frac{z^N}{(N-m)!}
         \int_{\Lambda^{N-m}}\!\! 
            e^{-\beta U_N(\RR_N)}
         \dRR_{m,N}\,,
\ee
where $\RR_m\in\Lambda^m$, $\RR_{m,N}$ denotes the $(N-m)$-tupel 
$(R_{m+1},\dots,R_N)\subset\Lambda^{N-m}$ with the
coordinates of additional $N-m$ particles, and
\be{Xi}
   \Xi_\Lambda \,=\, 
   \sum_{N=0}^\infty \frac{z^N}{N!}\int_{\Lambda^N} e^{-\beta U_N(\RR_N)}\dRR_N
\ee
is the \emph{grand canonical partition function} for a given inverse
temperature $\beta$ and activity $z$. The formulations in \req{rho-m}
and \req{Xi} obey the usual convention that the integral of the constant one 
over $\Lambda^0$ is considered to be one. 

It has been shown in \cite{Ruel69} that under the given assumptions on $u$
and for an inverse temperature $\beta$ and activity $z$ satisfying
\be{z}
   0 \,<\, z \,<\,
   \frac{1}{c_\beta e^{2\beta B+1}}
\ee
the distribution function $\rho_\Lambda^{(m)}$ is bounded in $\Lambda^m$, its
bound being independent of the size of $\Lambda$, and
in the \emph{thermodynamical limit} $|\Lambda|\to\infty$
the distribution functions converge compactly,
i.e., uniformly on every compact subset of $(\R^3)^m$;
we denote by $\rho^{(m)}:(\R^3)^m\to\R_0^+$ the corresponding limit function. 
When $m=1$ the resulting limit is the constant \emph{counting density} $\rho_0$
of the system; when $m=2$
the limit $\rho^{(2)}$ is invariant under translations and rotations and
one can define the \emph{radial distribution function} $g:\R^+\to\R^+_0$ by
\bdm
   g(r) \,=\, \frac{1}{\rho_0^2}\,\rho^{(2)}(R_1,R_2)\,, \qquad r=|R_1-R_2|\,.
\edm

In this work we investigate the dependence of the molecular distribution
functions on $u$ and we are going to prove the following two results.

\begin{result}
\label{Thm:A}
Let $u$ satisfy Assumption~\ref{Ass:u} and $z$ and $\beta$ be constrained
by \req{z}. Then for every $m\in\N$ and every box $\Lambda\subset\R^3$ 
the molecular distribution function $\rho_\Lambda^{(m)}$ has a well-defined 
Fr\'echet derivative 
$\partial\rho_\Lambda^{(m)}\in\L\bigl(\V,L^\infty(\Lambda^m)\bigr)$ with respect
to $u$, and $\rho^{(m)}$ has a Fr\'echet derivative
$\partial\rho^{(m)}\in\L\bigl(\V,L^\infty((\R^3)^m)\bigr)$.
\end{result}

\begin{result}
\label{Thm:B}
Under the assumptions of Theorem~\ref{Thm:A} the derivative operator
$\partial\rho_\Lambda^{(m)}$ converges to $\partial\rho^{(m)}$
in the thermodynamical limit $|\Lambda|\to\infty$
in the following sense: For any fixed bounded box $\Lambda'\subset\R^3$
and any $m\in\N$ there holds
\bdm
   \norm{(\partial\rho_\Lambda^{(m)})v - (\partial\rho^{(m)})v}_{{\Lambda'}^m}
   \,\to\, 0\,,
\edm
uniformly for $v\in\V$ with $\norm{v}_\V\leq 1$.
\end{result}

The proofs of these two theorems utilize the 
Kirkwood-Salsburg equations. As such, the Fr\'echet derivatives are only
given implicitly as the solution of a semi-infinite linear system.
For a bounded box $\Lambda\subset\R^3$ and associated molecular distribution
functions $\rho_\Lambda^{(m)}$ more explicit formulae can be derived directly 
from \req{rho-m},
and we will do so for $m=1$ and $m=2$ in Section~\ref{Sec:explicit}; 
we mention, though, that
it is much more difficult to investigate the thermodynamical limit of the
latter and to prove differentiability of $\rho^{(m)}$ by this approach.

\section{Proof of Theorem~\ref{Thm:A}}
\label{Sec:A}
In the sequel we consider a generic box $\Lambda\subset\R^3$ 
centered at the origin, and, by some abuse of notation, 
we will even allow $\Lambda$ to be the entire space.

To begin with we recall the definition~\req{jstar} of $j^*(\RR_m)$ for a given 
$\RR_m\in\Lambda^m$, and associated with it we introduce
the (nonlinear) projection $\Pi_m:(\R^3)^m\to(R^3)^{m-1}$ via
\be{Pim}
   \Pi_m : \RR_m \,\mapsto\, (\RR_{j^*-1},\RR_{j^*\!,m})\,.
\ee
We also need to define the following two function sequences
$d_m:(\R^3)^m\to\R$, $m\in\N$, and $k_n:\R^3\times(\R^3)^n\to\R$, $n\in\N$, 
given by
\be{dm}
   d_m(\RR_m) \,=\, \prod_{i=1\atop \,\,\,i\neq j^*}^m\!e^{-\beta u(|R_i-R_{j^*}|)}\,,
   \qquad m\in\N\,,
\ee
and 
\be{kernel}
   k_n(R;\RR'_n) 
   \,=\, \prod_{i=1}^n f(R_i'-R)\,,
\ee
where $f$ is the Mayer function~\req{Mayer}.
The latter are utilized as kernel functions for the integral operators
\be{K-nu-tilde}
   \bigl(K_{mn,\Lambda} \varphi_{m+n-1}\bigr)(\RR_m) 
   \,=\, \frac{1}{n!} 
         \int_{\Lambda^{n}}
            k_n(R_{j^*};\RR'_n) \varphi_{m+n-1}
            \bigl(\Pi_m(\RR_m),\RR'_n\bigr)\dRR'_n
\ee
for $m,n\in\N$, and $\varphi_{m+n-1}:\Lambda^{m+n-1}\to\R$.
Finally, we need the extension operators
\be{I-nu-prime}
   (I_{m,\Lambda}\varphi_{m-1})(\RR_m) \,=\, \varphi_{m-1}\bigl(\Pi_m(\RR_m)\bigr)
\ee
for $m\in\N\setminus\{1\}$ and $\varphi_{m-1}:\Lambda^{m-1}\to\R$.

Following Ruelle~\cite{Ruel69} we introduce
the Banach space $\X_\Lambda$ of sequences $\bfphi=(\varphi_m)_m$
of bounded functions $\varphi_m:\Lambda^m\to\R$, for which the norm
\bdm
     \bigl\|\bfphi\bigr\|_{\X_\Lambda} \,=\, 
   \sup_{m\in\N} \cbeta^m \norm{\varphi_m}_{\Lambda^m}
\edm
with $\cbeta$ of \req{cbeta} is finite.
On $\X_\Lambda$ a diagonal multiplication operator
\be{DK-D}
   D_\Lambda \,=\, \begin{cmatrix} 
                    d_1 & \\
                    & d_2 & \\
                    & & \ddots & 
                 \end{cmatrix}
\ee
and a semi-infinite block integral operator
\be{DK-K}
   K_\Lambda \,=\, \begin{cmatrix} 
                    K_{11,\Lambda} & K_{12,\Lambda} & K_{13,\Lambda} &
                           \dots\phantom{\vdots} \\
                    I_{2,\Lambda}  & K_{21,\Lambda} & K_{22,\Lambda} &
                           \vphantom{\vdots} \\
                    0 & I_{3,\Lambda} & K_{31,\Lambda} & \ddots \\
                    \vdots & \ddots & \ddots & \ddots
           \end{cmatrix}
\ee
are defined by utilizing the functions $d_m$ of \req{dm} and the operators
$K_{mn,\Lambda}$ and $I_{m,\Lambda}$ of \req{K-nu-tilde} and \req{I-nu-prime},
respectively.
It has been shown in \cite{Ruel69} that $D_\Lambda$ and $K_\Lambda$ belong to
$\L(\X_\Lambda)$ with
\be{Ruelle-KS}
   \norm{D_\Lambda}_{\L(\X_\Lambda)} \,\leq\, e^{2\beta B}\,, \qquad
   \norm{K_\Lambda}_{\L(\X_\Lambda)} \,\leq\, \cbeta e\,.
\ee
These operators can now be used to formulate the celebrated Kirkwood-Salsburg 
equations:
\be{KS-prime}
   (I-zA_\Lambda) \bfrho_\Lambda \,=\, z\bfe\,, \qquad 
   A_\Lambda=D_\Lambda K_\Lambda\,,
\ee
where
\bdm
   \bfrho_\Lambda = (\rho^{(m)}_\Lambda)_m 
   \qquad \text{and} \qquad
   \bfe \,=\, (\delta_{1m})_m\,,
\edm
with $\delta_{1m}$ being the Kronecker symbol and $z$ being the activity.
By virtue of \req{z} the linear system \req{KS-prime} is uniquely solvable
for $\bfrho_\Lambda$.

To compute the derivative of the molecular distribution functions we need
to determine the derivatives of the functions $d$, $f$, and $k$ in 
appropriate topologies. This is the purpose of the following three auxiliary
results.

\begin{lemma}
\label{Lem:Mayer-derivative}
Let $u$ satisfy Assumption~\ref{Ass:u}. Then the Mayer $f$-function has a 
Fr\'echet derivative $\partial f\in\L\bigl(\V,L^1(\R^3)\bigr)$ with respect 
to $u$ given by
\bdm
   (\partial f)v \,=\, -\beta e^{-\beta u}v\,, 
   \qquad v\in\V\,,
\edm
i.e., there exists $\Cbeta>0$ such that
\begin{subequations}
\label{eq:Lem:Mayer-derivative-all}
\begin{align}
\label{eq:Lem:Mayer-derivative-bound}
   \norm{(\partial f)v}_{L^1(\R^3)} &\,\leq\, \Cbeta \norm{v}_{\V}\,,
\intertext{and for $\norm{v}_\V\leq t_0$
and $\ftilde$ the Mayer function associated with $\utilde=u+v$ there holds}
\label{eq:Calpha}
   \norm{\ftilde - f}_{L^1(\R^3)} &\,\leq\, \Cbeta \norm{v}_{\V}\,, \\[1ex]
\label{eq:Lem:Mayer-derivative}
   \norm{\ftilde - f - (\partial f)v}_{L^1(\R^3)} 
   &\,\leq\, \Cbeta \norm{v}_{\V}^2\,.
\end{align}
\end{subequations}
\end{lemma}

\begin{proof}
For $v\in\V$ Taylor's theorem yields
\begin{align}
\nonumber
   \norm{\ftilde - f - (\partial f)v}_{L^1(\R^3)}
   &\,=\, 4\pi \int_0^\infty 
          \bigl| e^{-\beta\utilde(r)} - e^{-\beta u(r)} 
                 + \beta e^{-\beta u(r)} v(r)\bigr|\,r^2\!\dr \\[1ex]
\nonumber
   &\,=\, 4\pi \int_0^\infty e^{-\beta u(r)} \bigl| e^{-\beta v(r)} - 1 
                                          + \beta v(r)\bigr|\,r^2\!\dr \\[1ex]
\label{eq:Mayer-derivative-tmp}
   &\,\leq\, 2\pi \int_0^\infty 
                  e^{-\beta(u(r)-|v(r)|)}|\beta v(r)|^2 r^2\!\dr\,.
\end{align}
When $\norm{v}_{\V}\leq t_0$ then it follows from Proposition~\ref{Prop:U} that
\begin{align}
\nonumber
   e^{-\beta(u(r)-|v(r)|)}|\beta v(r)|^2 r^2
   &\,\leq\,\beta^2 e^{2\beta B} \norm{v}_{\V}^2 \uo(s)\, \uo(r) r^2\\
\intertext{for $r>\rs$, while for $0<r\leq \rs$ we have}
\nonumber
   e^{-\beta(u(r)-|v(r)|)}|\beta v(r)|^2 r^2
   &\,\leq\, \beta^2\norm{v}_{\V}^2 
   e^{-\beta(1-t_0)u(r)} u^2(r) s^2\\
   \nonumber
   &\,\leq\, \frac{4\rs^2}{e^2(1-t_0)^2}\,\norm{v}_{\V}^2\,.
\end{align}
Inserting these estimates into \req{Mayer-derivative-tmp} we readily arrive
at \req{Lem:Mayer-derivative}.

In much the same way we can estimate
\be{estimate2-tmp}
   e^{-\beta u(r)} |\beta v(r)| r^2 \,\leq\,
   \begin{cases}
      \beta e^{2\beta B} \norm{v}_{\V} \uo(r) r^2\,, & 
      r\geq \rs\,, \\
      (\rs^2/e) \norm{v}_{\V}\,, & 
      0 < r \leq r_0\,,
   \end{cases}
\ee
to deduce \req{Lem:Mayer-derivative-bound}.
Take note that \req{estimate2-tmp} holds true for every $v\in\V$.

When $u$ is replaced by $u-|v|$ with $\norm{v}_\V\leq t_0$ in
\req{estimate2-tmp} then the upper bound on the right-hand side increases
by at most $1/(1-t_0)$, and hence, 
\req{estimate2-tmp} also provides a convergent majorant to estimate
\bdm
   \norm{\ftilde - f}_{L^1(\R^3)}
   \,=\, 4\pi\!\int_0^\infty \!
         e^{-\beta u(r)}\bigl|e^{-\beta v(r)}-1\bigr|\,r^2\!\dr
   \,\leq\, 4\pi\!\int_0^\infty\!
            e^{-\beta (u(r)-|v(r)|)}|\beta v(r)|\,r^2\!\dr
\edm
for $\norm{v}_\V\leq t_0$, cf.~\req{Calpha}.
\end{proof}

\begin{lemma}
\label{Lem:d-derivative}
Let $u$ satisfy Assumption~\ref{Ass:u}. Then
the functions $d_m$ of \req{dm} are Fr\'echet differentiable with
respect to the pair potential with derivative 
$\partial d_m\in\L\bigl(\V,L^\infty(\R^3)\bigr)$ given by
\bdm
   \bigl((\partial d_m)v\bigr)(\RR_m)
   \,=\, -\beta d_m(\RR_m) 
         \!\sum_{i=1\atop \,\,\,i\neq j^*}^m \!v(|R_i-R_{j^*}|)
\edm
for $v\in\V$ and $\RR_m\in(\R^3)^m$. There holds
\begin{subequations}
\label{eq:Lem:d-derivative-all}
\begin{align}
\label{eq:Lem:d-derivative}
   \norm{(\partial d_m)v}_{(\R^3)^m} 
   &\,\leq\, \frac{e^{2\beta B}}{t_0}\,\norm{v}_{\V}
\intertext{and, if $\norm{v}_{\V}\leq t_0/2$ then}
\label{eq:dm-diff}
   \norm{\dtilde_m-d_m}_{(\R^3)^m} 
   &\,\leq\, \frac{2e^{2\beta B}}{t_0}\,\norm{v}_{\V}\,,\\[1ex]
\label{eq:18.x}
   \norm{\dtilde_m - d_m - (\partial d_m)v}_{(\R^3)^m}
   &\,\leq\, \frac{4e^{2\beta B}}{t_0^2}\,\norm{v}_{\V}^2 \,,
\end{align}
\end{subequations}
where $\dtilde_m$ denotes the function $d_m$ associated with the pair potential
$u+v$.
\end{lemma}

\begin{proof}
For the proof of this result it is essential that the index $j^*=j^*(\RR_m)$
of \req{jstar} is independent of the particular pair potential $\utilde=u+v$
as long as $\norm{v}_\V\leq t_0$.
For $v\in\V$ and fixed $\RR_m\in(\R^3)^m$ we use Taylor's theorem to estimate
\bdm
   \bigl|\dtilde_m - d_m - (\partial d_m)v\bigr| 
   \,=\, \bigl| d_m(e^{-\beta \sum v} - 1 + \beta\sum v) \bigr|
   \,\leq\, \frac{1}{2}\,\bigl| e^{-\beta\sum u} (\beta\sum v)^2 
                                e^{\beta\sum |v|}\bigr|\,,
\edm 
where all sums extend over every $i=1,\dots,m$ with $i\neq j^*$,
the respective arguments being given by $|R_i-R_{j^*}|$.
For every $\gamma>0$ there holds
\bdm
   \frac{1}{2}\,(\beta\sum v)^2
   \,\leq\, \gamma^2\,\frac{1}{2}\bigl(\frac{\beta}{\gamma}\sum |v|\bigr)^2
   \,\leq\, \gamma^2 e^{\beta\sum |v|/\gamma}\,,
\edm
and hence
\bdm
   \bigl|\dtilde_m - d_m - (\partial d_m)v\bigr| 
   \,\leq\, \gamma^2 e^{-\beta\sum (u - |v| - |v|/\gamma)}\,.
\edm
For $\gamma=2\norm{v}_{\V}/t_0$ 
the assertion~\req{18.x} now follows from \req{jstar} with 
$\utilde = u - |v| - |v|/\gamma$
provided that $\norm{v}_{\V}\leq t_0/2$.
This proves the differentiability of $d_m$. 
The proof of the two estimates \req{Lem:d-derivative} and \req{dm-diff}
follows along the same lines and is left to the reader. 
\end{proof}

To formulate the following result we introduce
the formal derivative of $k_n$ of \req{kernel} with respect to $u$ in
direction $v$, i.e.,
\be{kprime}
   k_n'(R;\RR_n') \,=\, 
   \sum_{i=1}^n \bigl((\partial f)v\bigr)(|R_i'-R|)
                       \prod_{j=1\atop j\neq i}^n f(R_j'-R)
\ee
for $R\in\R^3$ and $\RR_N'\in(\R^3)^n$.

\begin{proposition}
\label{Prop:K-derivative}
Under the assumptions of Theorem~\ref{Thm:A}
the operator $K_\Lambda$ is Fr\'echet differentiable with respect to $u$ in 
$\L\bigl(\V,\L(\X_\Lambda)\bigr)$, and its derivative $\partial K_\Lambda$ is
given by
\be{DKprime}
   (\partial K_\Lambda)v
   \,=\, \begin{cmatrix} 
            K_{11,\Lambda}' & K_{12,\Lambda}' & K_{13,\Lambda}' &
                      \dots\phantom{\vdots} \\
            0       & K_{21,\Lambda}' & K_{22,\Lambda}' & \vphantom{\vdots} \\
            \vdots  & 0       & K_{31,\Lambda}' & \ddots \\
            \vdots  &         & \ddots  & \ddots
         \end{cmatrix},
\ee
where
\bdm
   (K_{mn,\Lambda}'\varphi_{m+n-1})(\RR_m)
   \,=\, \frac{1}{n!} 
         \int_{\Lambda^n} k_n'(R_{j^*};\RR_n')
            \varphi_{m+n-1}\bigl(\Pi_m(\RR_m),\RR_n'\bigr)\dRR_n'
\edm
for $m\in\N$, with $\varphi_{m+n-1}:\Lambda^{m+n-1}\to\R$, $\Pi_m$ of \req{Pim},
and $k_n'$ of \req{kprime}.
\end{proposition}

\begin{proof}
Let $\Y_n$ be the Banach space of functions $k:\R^3\times(\R^3)^n\to\R$ 
with norm
\bdm
   \norm{k}_{\Y_n}
   \,=\, \sup_{R\in\R^3} \norm{k(R;\,\cdot\,)}_{L^1((\R^3)^n)}\,.
\edm
We prove that $k_n$ of \req{kernel} is differentiable with respect to $u$
in $\L(\V,\Y_n)$, and that its derivative $(\partial k_n)v$ in direction $v$
is given by $k_n'$ of \req{kprime}. 
From Lemma~\ref{Lem:Mayer-derivative} we readily conclude that
$k_n'\in\Y_n$ with
\be{kprimenorm}
   \norm{k_n'}_{\Y_n}
   \,\leq\, n \Cbeta \cbeta^{n-1} \norm{v}_{\V}\,.
\ee

We now assume that $\norm{v}_{\V}\leq t_0$ and
introduce short-hand notations $\ktilde_n$ for the kernel $k_n$ associated
with the pair potential $\utilde=u+v$, and let
\bdm
   f_i(R) \,=\, e^{-\beta u(|R_i'-R|)}-1 \qquad \text{and} \qquad
   \ftilde_i(R) \,=\, e^{-\beta \utilde(|R_i'-R|)}-1
\edm
for every appropriate value of $i$; in the sequel we also omit the obvious
arguments $R_i'$ and $R$, respectively.
Then we prove by induction on $n\in\N$ that
\be{kn-induction1}
   \norm{\ktilde_n-k_n}_{\Y_n}
   \,\leq\, n\Cbeta \cbeta^{n-1}\,\norm{v}_{\V}\,,\qquad 
   n\in\N\,,
\ee
which is obviously true because of \req{Calpha} when $n=1$, 
because $\ktilde_1-k_1=\ftilde_1-f_1$. The inductive step is then based on
\bdmal
   \ktilde_{n+1}-k_{n+1}
   &\,=\, \ktilde_n \ftilde_{n+1} - k_nf_{n+1}
    \,=\, \ktilde_n(\ftilde_{n+1}-f_{n+1}) \,+\, (\ktilde_n-k_n)f_{n+1}\\[1ex]
   &\,=\, \ftilde_1\cdots\ftilde_n(\ftilde_{n+1}-f_{n+1})
          \,+\, (\ktilde_n-k_n)f_{n+1}\,,
\edmal
the induction hypothesis, and on the estimates \req{Calpha} and \req{cbeta}.

Still assuming that $\norm{v}_{\V}\leq t_0$ we now
proceed to establish the inequality
\be{kn-induction2}
   \norm{\ktilde_n - k_n - k_n'}_{\Y_n}
   \,\leq\, n^2 \Cbeta' \cbeta^{n-1} \norm{v}^2_{\V}\,, \qquad 
   n\in\N\,,
\ee
for $\Cbeta'=\max\{\Cbeta,\Cbeta^2/(2\cbeta)\}$, which proves the 
asserted differentiability of $k_n$.
For $n=1$ inequality~\req{kn-induction2} has been established in the proof of
Lemma~\ref{Lem:Mayer-derivative}, cf.~\req{Lem:Mayer-derivative}.
For the induction step we write
\bdmal
   \ktilde_{n+1} - k_{n+1} - k_{n+1}'
   &\,=\, \ktilde_n\ftilde_{n+1} - k_nf_{n+1} - k_n(\partial f_{n+1})v_{n+1}
          - k_n'f_{n+1}\\[1ex]
   &\,=\, (\ktilde_n-k_n-k_n')f_{n+1}
          + (\ktilde_n-k_n)(\ftilde_{n+1}-f_{n+1})\\
   &\phantom{\,=\, (\ktilde_n-k_n-k_n')f_{n+1}}\
        + f_1\cdots f_n(\ftilde_{n+1}-f_{n+1}-(\partial f_{n+1})v_{n+1})\,,
\edmal
where we have set
\bdm
   v_{n+1} \,=\, v(|R_{n+1}'-R|)\,.
\edm
Then it follows from the induction hypothesis~\req{kn-induction2}, 
and from \req{cbeta}, \req{kn-induction1}, \req{Calpha}, and 
\req{Lem:Mayer-derivative} that
\bdm
   \norm{\ktilde_{n+1}-k_{n+1}-k_{n+1}'}_{\Y_{n+1}}
   \,\leq\, \bigl(n^2\Cbeta'\cbeta^n \,+\, n\Cbeta^2 \cbeta^{n-1}
                  \,+\, \cbeta^n \Cbeta
            \bigr)\norm{v}^2_{\V}\,.
\edm
Inserting the definition of $\Cbeta'$ we further conclude that
the above right-hand side satisfies
\bdm
   \bigl(n^2\Cbeta' \,+\, 2n\,\frac{\Cbeta^2}{2\cbeta}\,+\, \Cbeta
   \bigr)\cbeta^n\norm{v}^2_{\V}
   \,\leq\, (n+1)^2\Cbeta'\cbeta^n\norm{v}^2_{\V}\,,
\edm
hence the induction step is complete.

Having established \req{kn-induction2} we can now argue as in \cite{Ruel69}
to show that
\bdm
   \norm{\Ktilde_\Lambda - K_\Lambda - (\partial K_\Lambda)v}_{\L(\X_\Lambda)}
   \,\leq\, 2e\Cbeta'\norm{v}_{\V}^2\,,
\edm
where $\Ktilde_\Lambda$ denotes the block integral operator~\req{DK-K}
associated with $\utilde$, $(\partial K_\Lambda)v$ is defined in \req{DKprime},
and $\norm{v}_{\V}\leq t_0$. This shows that $\partial K_\Lambda$ of \req{DKprime}
is the Fr\'echet derivative of $K_\Lambda$ when considered a function of the
pair potential.
\end{proof}

Now we can establish Theorem~\ref{Thm:A}.

\begin{xproof}{\ref{Thm:A}}
From Lemma~\ref{Lem:d-derivative} it follows readily
that $D_\Lambda$ is also Fr\'echet differentiable as a
function of $u$, and so is $A_\Lambda=D_\Lambda K_\Lambda$,
and the derivative $\partial A_\Lambda$ of $A_\Lambda$ is given by
\bdm
   (\partial A_\Lambda)v =
   ((\partial D_\Lambda)v) K_\Lambda + D_\Lambda(\partial K_\Lambda)v
\edm
for $v\in\V$. Denote by $\Dtilde_\Lambda$ and $\Ktilde_\Lambda$ 
the operators~\req{DK-D} and \req{DK-K} associated with the pair potential
$\utilde = u+v$, and set $\Atilde_\Lambda=\Dtilde_\Lambda\Ktilde_\Lambda$. 
Moreover, let $\bfrhotilde_\Lambda$ be the respective 
sequence of molecular distribution functions.
Then, for $\norm{v}_\V\leq t_0$ there holds
\bdmal
   \bfrhotilde_\Lambda - \bfrho_\Lambda
   &\,=\, \bigl((I-z\Atilde_\Lambda)^{-1} - (I-zA_\Lambda)^{-1}\bigr)z\bfe\\
   &\,=\, z\bigl(I-z\Atilde_\Lambda)^{-1}
          (\Atilde_\Lambda-A_\Lambda)(I-zA_\Lambda)^{-1}z\bfe\\
   &\,=\, z(I-zA_\Lambda)^{-1}
          \bigl((\partial A_\Lambda)v\bigr)(I-zA_\Lambda)^{-1}z\bfe
          \,+\, O(\norm{v}^2_{\V})
\edmal
in $\X_\Lambda$.
Accordingly, $\bfrho_\Lambda$ is Fr\'echet differentiable, and its 
Fr\'echet derivative $\partial\bfrho_\Lambda\in\L(\V,\X_\Lambda)$ satisfies
\bdm
   (\partial\bfrho_\Lambda)v 
   \,=\, z(I-zA_\Lambda)^{-1}\bigl((\partial A_\Lambda)v\bigr)(I-zA_\Lambda)^{-1}z\bfe
   \,=\, z(I-zA_\Lambda)^{-1}\bigl((\partial A_\Lambda)v\bigr)\bfrho_\Lambda\,.
\edm
This implies the statement of Theorem~\ref{Thm:A}.
\end{xproof}

\section{Proof of Theorem~\ref{Thm:B}}
\label{Sec:B}
In this section we investigate the thermodynamical limit of the
sequence $\bfrho_\Lambda=(\rho_\Lambda^{(m)})_m\in\X_\Lambda$
of molecular distribution functions; throughout this section 
$\Lambda\subset\R^3$ always denotes a finite size box.
Take note that when $\Lambda'\subset\Lambda$ then any sequence
$\bfphi_\Lambda=(\varphi_\Lambda^{(m)})_m\in\X_\Lambda$ also belongs to
$\X_{\Lambda'}$. We will make repeated use of this property without further
mentioning by taking the corresponding norm $\norm{\bfphi_\Lambda}_{\X_{\Lambda'}}$
of $\bfphi_\Lambda$.
Often, however, it will be convenient to make this imbedding more explicit.
To this end we introduce the corresponding imbedding operator
\be{P}
   P_{\Lambda'}\,:\, \X_\Lambda \to \X_{\Lambda'}\,, \qquad
   P_{\Lambda'}\bfphi_\Lambda =\bigl(\varphi_\Lambda^{(m)}\big|_{(\Lambda')^m}\bigr)_m\,.
\ee
To simplify matters, though, we will not specify explicitly the domain 
$\X_{\Lambda}$ of $P_{\Lambda'}$ in this notation; this domain may differ depending
on the context.

In the previous section we have shown that
$\bfrho_\Lambda$ as well as its thermodynamical limit $\bfrho$ 
are differentiable with respect to $u$, their derivative(s)
being given by
\be{rho-derivatives}
\begin{aligned}
   (\partial\bfrho_\Lambda)v 
   &\,=\, z(I-zA_\Lambda)^{-1}A_\Lambda'\bfrho_\Lambda\,,
   \\
   (\partial\bfrho)v 
   &\,=\, z(I-zA)^{-1}A'\bfrho\,,
\end{aligned}
\ee
respectively. 
In \req{rho-derivatives} we have used short-hand notations $A'$ and
$A_\Lambda'$ for $(\partial A)v$ and $(\partial A_\Lambda)v$, respectively.
We will continue to do so and similarly for
the operators $D$, $K$, $D_\Lambda$, and $K_\Lambda$
throughout the remainder of this section 
as long as $v\in\V$ is fixed.

The proof of Theorem~\ref{Thm:B} will proceed in three steps:
First, in Lemma~\ref{Lem:Aprimelimit}, we will show that
for two boxes $\Lambda''\subset\Lambda$ with $\Lambda''$ being kept fixed
we have convergence
\be{Lem:Aprimelimit}
   \norm{A_\Lambda'\bfrho_\Lambda - A'\bfrho}_{\X_{\Lambda''}} \,\to\, 0 \qquad
   \text{as $|\Lambda|\to\infty$}\,.
\ee
Then, in a second step we fix another box $\Lambda'\subset\Lambda$, 
and we argue that
\be{term-I}
   \norm{(I-zA_\Lambda)^{-1}P_\Lambda A'\bfrho
         \,-\, (I-zA)^{-1}A'\bfrho}_{\X_{\Lambda'}}
   \,\to\, 0 \qquad \text{as $|\Lambda|\to\infty$}\,.
\ee
Third, we apply \req{Lem:Aprimelimit} in a setting with three boxes
$\Lambda'\subset\Lambda''\subset\Lambda$ to show that
\be{term-II}
   \norm{(I-zA_\Lambda)^{-1}P_\Lambda
         (A'\bfrho-A_\Lambda'\bfrho_\Lambda)}_{\X_{\Lambda'}}
   \,\to\, 0 \qquad \text{as $|\Lambda|\to\infty$}\,,
\ee
provided that $\Lambda'$ is kept fixed.
A combination of \req{term-I} and \req{term-II} then readily yields
the desired convergence
\bdm
   \norm{(\partial\bfrho_\Lambda)v-(\partial\bfrho)v}_{\X_{\Lambda'}} \,\to\,0
\edm
as $|\Lambda|\to\infty$, which completes the proof of Theorem~\ref{Thm:B}.

\begin{lemma}
\label{Lem:Aprimelimit}
Let $u$ satisfy the assumptions of Theorem~\ref{Thm:A}
and let $A_\Lambda'$ and $A'$ be defined as above for a given
$v\in\V$. Then, for a fixed box $\Lambda''\subset\R^3$ and for
$\Lambda''\subset\Lambda\subset\R^3$ the convergence \req{Lem:Aprimelimit}
holds true
uniformly for all $v\in\V$ with $\norm{v}_\V\leq 1$.
\end{lemma}

\begin{proof}
Let $\eps>0$, $d>s$, and a certain box $\Lambda''\subset\R^3$ be given.
We choose a second box $\Lambda^*\supset\Lambda''$
in such a way that $|R''-R|>d$ for every
$R''\in\Lambda''$ and every $R\in\R^3\setminus\Lambda^*$.
According to \cite{Ruel69} there holds
\bdm
   \norm{\bfrho_\Lambda - \bfrho}_{\X_{\Lambda^*}} \,\leq\,\eps\,,
\edm
provided that $\Lambda\supset\Lambda^*$ is sufficiently large; moreover,
\bdm
   \norm{\bfrho}_{\X_{\R^3}}, 
   \norm{\bfrho_\Lambda}_{\X_{\Lambda}}
   \,\leq\, C_*
\edm
for some $C_*>0$, independent of the size of $\Lambda$.
Then, for any $\RR_m\in(\Lambda'')^m$ and any $n\in\N$ we can estimate
\begin{align}
\nonumber
   &\!\!\!\!\bigl|
      \bigl(K_{mn,\Lambda}\rho_\Lambda^{(m+n-1)}\bigr)(\RR_m)
           \,-\,\bigl(K_{mn}\rho^{m+n-1}\bigr)(\RR_m)
   \bigr|\\[1ex]
\nonumber
   &\!\!\!\!
    \,=\, \frac{1}{n!}\,\Bigg|
            \int_{\Lambda^n} \!\!\!
                k_n(R_{j^*};\RR_n')\rho^{(m+n-1)}_\Lambda(\Pi_m(\RR_m),\RR_n')
            \dRR_n'
            \\
\nonumber
   &\!\!\!\! \qquad \qquad \qquad
            \,-\,
            \int_{(\R^3)^n} \!\!\!
                k_n(R_{j^*};\RR_n')\rho^{(m+n-1)}(\Pi_m(\RR_m),\RR_n')\dRR_n'
          \Bigg|\\[1ex]
\nonumber
   &\!\!\!\!
    \,=\, \frac{1}{n!}\,\Bigg|
            \int_{(\Lambda^*)^n} \!\!\!
               k_n(R_{j^*};\RR_n')
                  \bigl(\rho^{(m+n-1)}_\Lambda(\Pi_m(\RR_m),\RR_n')
                                 - \rho^{(m+n-1)}(\Pi_m(\RR_m),\RR_n')\bigr)
            \dRR_n'
            \\
\nonumber
   &\!\!\!\! \qquad \qquad \qquad
            \,+\,
            \int_{\Lambda^n\setminus(\Lambda^*)^n} \!\!\!
                k_n(R_{j^*};\RR_n')\rho^{(m+n-1)}_\Lambda(\Pi_m(\RR_m),\RR_n')
            \dRR_n'\\
\nonumber
   &\!\!\!\! \qquad \qquad \qquad
            \,-\,
            \int_{(\R^3)^n\setminus(\Lambda^*)^n} \!\!\!
                k_n(R_{j^*};\RR_n')\rho^{(m+n-1)}(\Pi_m(\RR_m),\RR_n')\dRR_n'
          \Bigg|\\[1ex]
\label{eq:Lem:Aprimelimit-tmp}
   &\!\!\!\!
    \,\leq\, \frac{\cbeta^{1-m-n}}{n!}\Bigl(\eps\!
             \int_{(\R^3)^n}\!\bigl|k_n(R_{j^*};\RR_n')\bigr|\dRR_n'
             + 
             2C_*\!
             \int_{(\R^3)^n\setminus(\Lambda^*)^n} \!
                \bigl|k_n(R_{j^*};\RR_n')\bigr|\dRR_n'\Bigr)\,.
\end{align}
The first integral in \req{Lem:Aprimelimit-tmp} is bounded by $\cbeta^n$, the
second integral can be estimated by $n\cbeta^{n-1}\cbetad$,
where
\be{cbetad}
   \cbetad \,=\, 4\pi\int_d^\infty |e^{-\beta \utilde(r)}-1|\,r^2\!\dr \,,
\ee
compare the proof of \cite[Theorem~4.2.3]{Ruel69}. We thus have
\bdm
   \cbeta^m
   \norm{K_{mn,\Lambda}\rho_\Lambda^{(m+n-1)}-
         K_{mn}\rho^{(m+n-1)}}_{(\Lambda'')^m}
   \,\leq\, \frac{\cbeta}{n!}\,\eps
            \,+\, \frac{2C_*}{(n-1)!} \, \cbetad\,,
\edm
which readily yields
\bdmal
   \norm{K_\Lambda\bfrho_\Lambda - K\bfrho}_{\X_{\Lambda''}}
    \,\leq\, e(\cbeta \eps + 2C_* \cbetad)\,.
\edmal
Since we can let $d$ be arbitrarily large and $\eps>0$ be arbitrarily small
by choosing $|\Lambda^*|$ and $|\Lambda|$ sufficiently big, it follows that
\bdm
   \norm{K_\Lambda\bfrho_\Lambda - K\bfrho}_{\X_{\Lambda''}}
   \,\to\, 0 \qquad \text{as $|\Lambda|\to\infty$}\,,
   \edm
and hence, cf.~\req{Lem:d-derivative}, 
\bdmal
   &\norm{D_\Lambda'K_\Lambda\bfrho_\Lambda - D'K\bfrho}_{\X_{\Lambda''}}
    \,=\, \norm{D_{\Lambda''}'P_{\Lambda''}
                (K_\Lambda\bfrho_\Lambda-K\bfrho)}_{\X_{\Lambda''}}\\[1ex]
   &\qquad
    \,\leq\, \norm{D'_{\Lambda''}}_{\X_{\Lambda''}} 
             \norm{K_\Lambda\bfrho_\Lambda - K\bfrho}_{\X_{\Lambda''}}
    \,\leq\, \frac{e^{2\beta B}}{t_0}\,
             \norm{K_\Lambda\bfrho_\Lambda - K\bfrho}_{\X_{\Lambda''}}
             \norm{v}_{\V}
    \,\to\, 0
\edmal
as $|\Lambda|\to\infty$, uniformly for all $v\in\V$ with $\norm{v}_\V\leq 1$.

In a similar way we can estimate
\bdm
   \norm{D_\Lambda K_\Lambda'\bfrho_\Lambda 
         - DK'\bfrho}_{\X_{\Lambda''}}\,,
\edm
the main difference for this estimate being that $k_n$ in 
\req{Lem:Aprimelimit-tmp} has to be replaced by $k_n'$ of \req{kprime}, i.e.,
\be{Kprimelimit}
\begin{aligned}
   &\!\!\!\!\bigl|\bigl(K_{mn,\Lambda}'
                        \rho_\Lambda^{(m+n-1)}\bigr)(\RR_m)
           \,-\,\bigl(K_{mn}'\rho^{m+n-1}\bigr)(\RR_m)
    \bigr|\\[1ex]
   &\!\!\!\!\
    \leq\, \frac{\cbeta^{1-m-n}}{n!}\Bigl( 
             \eps\!\int_{(\R^3)^n}\!\!\bigl|k_n'(R_{j^*};\RR_n')\bigr|\!\dRR_n'
             \,+\, 
             2C_*\!\int_{(\R^3)^n\setminus(\Lambda^*)^n} \!\!
                   \bigl|k_n'(R_{j^*};\RR_n')\bigr|\!\dRR_n'\Bigr)
\end{aligned}
\ee
for $\RR_m\in(\Lambda'')^m$.
Here we can use the bound \req{kprimenorm} for the first integral, 
but the second integral requires a different estimate.
Following the line of argument employed previously, we obtain from \req{kprime}
the inequality
\bdmal
   &\int_{(\R^3)^n\setminus(\Lambda^*)^n} \!\!
                \bigl|k_n'(R_{j^*};\RR_n')\bigr|\!\dRR_n'\\[1ex]
   &\ \,\leq\, \sum_{i=1}^n 
             \int_{(\R^3)^n\setminus(\Lambda^*)^n} \!
             \bigl|\bigl((\partial f)v\bigr)(|R_i'-R_{j^*}|)\bigr|\,
                \prod_{j\neq i} \,\bigl|f(|R_j'-R_{j^*}|)\bigr|
             \dRR_n'\\[1ex]
   &\ \,\leq\, \sum_{i=1}^n 
               \Biggl( (n-1) 
                       \int_{\R^3} 
                       \bigl|\bigl((\partial f)v\bigr)(|R|)\bigr|\!\dR \
                       \Bigl(\int_{\R^3} \bigl|f(|R|)\bigr|\!\dR
                       \Bigr)^{n-2}\!\!
                       \int_{\R^3\setminus\Lambda^*} 
                          \bigl|f(|R-R_{j^*}|)\bigr|\!\dR \\[1ex]
   &\qquad \qquad
      +\, \int_{\R^3\setminus\Lambda^*} 
             \bigl| \bigl((\partial f)v\bigr)(|R-R_{j^*}|)\bigr|\!\dR \
                 \Bigl(\int_{\R^3} \bigl|f(|R|)\bigr|\!\dR
                 \Bigr)^{n-1}
               \Biggr)\,.
\edmal
Using \req{cbeta}, \req{cbetad}, and \req{Lem:Mayer-derivative-bound}
this yields 
\bdmal
   &\int_{(\R^3)^n\setminus(\Lambda^*)^n} \!\!
               \bigl|k_n'(R_{j^*};\RR_n')\bigr|\dRR_n' \\
   &\qquad \,\leq\, \,n(n-1) \cbeta^{n-2}\cbetad\Cbeta \norm{v}_{\V}
   +\, n \cbeta^{n-1} \!\!\int_{|R|>d} 
             \bigl| \bigl((\partial f)v\bigr)(|R|)\bigr|\!\dR \,.
\edmal
Since $d>s$ the remaining integral can be estimated by means of
\req{estimate2-tmp}, and hence,
\bdmal
   \int_{|R|>d} 
   \bigl| \bigl((\partial f)v\bigr)(|R|)\bigr|\!\dR
   &\,=\,4\pi \int_d^\infty \beta e^{-\beta u(r)}\bigl|v(r)\bigr|
                               \,r^2\dr \\[1ex]
   &\,\leq\,4\pi \beta e^{2\beta B} \norm{v}_{\V}
            \int_d^\infty \uo(r)\,r^2\dr
    \,=:\, \go \norm{v}_{\V}\,.
    \edmal
Take note that $\go\to 0$ for $d\to\infty$.
We thus conclude from \req{Kprimelimit} that
\bdmal
   &\cbeta^m\,
    \bigl\|\bigl(K_{mn,\Lambda}'\rho_\Lambda^{(m+n-1)}\bigr)(\RR_m)
           \,-\,\bigl(K_{mn}'\rho^{m+n-1}\bigr)(\RR_m)
    \bigr\|_{(\Lambda'')^m}\\[1ex]
   &\qquad
    \leq\, \Bigl(\frac{\Cbeta}{(n-1)!}\,\eps\,+\,
                 \frac{2C_* \Cbeta}{\cbeta} \,\frac{n-1}{(n-1)!}\,\cbetad
                 \,+\, \frac{2C_*}{(n-1)!}\, \go 
           \Bigr)\norm{v}_{\V}\,.
\edmal
In view of \req{Ruelle-KS} it follows that
\bdmal
   &\norm{D_\Lambda K_\Lambda'\bfrho_\Lambda 
          - DK'\bfrho}_{\X_{\Lambda''}}
    \,=\, \norm{D_{\Lambda''}P_{\Lambda''}
            (K_\Lambda'\bfrho_\Lambda - K'\bfrho)}_{\X_{\Lambda''}}\\[1ex]
   &\,\leq\, \norm{D_{\Lambda''}}_{\X_{\Lambda''}}
             \norm{K_\Lambda'\bfrho_\Lambda 
                   - K'\bfrho}_{\X_{\Lambda''}}
    \,\leq\, e^{2\beta B+1}
             \Bigl(\Cbeta\eps\,+\,\frac{2C_* \Cbeta}{\cbeta} \,\cbetad
                   \,+\, 2C_* \go \Bigr)\norm{v}_{\V}\,,
\edmal
which can be made arbitrarily small by choosing $|\Lambda^*|$ and $|\Lambda|$
sufficiently big, uniformly for all $v\in\V$ with $\norm{v}_\V\leq 1$.

Since 
\bdm
   A_\Lambda'\bfrho_\Lambda - A'\bfrho \,=\, 
   (D_\Lambda'K_\Lambda\bfrho_\Lambda - D'K\bfrho)
   \,+\,(D_\Lambda K_\Lambda'\bfrho_\Lambda-DK'\bfrho)
\edm
we thus have established the assertion.
\end{proof}

As mentioned before another ingredient to the proof of Theorem~\ref{Thm:B}
is the statement~\req{term-I}, which we reformulate here in a slightly
more precise way.

\begin{lemma}
\label{Lem:term-I}
Let $u$ satisfy the assumptions of Theorem~\ref{Thm:A}
and let $A'$ be defined as above for a given $v\in\V$.
Then, for a fixed box $\Lambda'\subset\R^3$ and for
$\Lambda'\subset\Lambda\subset\R^3$ the convergence
\bdm
   \norm{(I-zA_\Lambda)^{-1}P_\Lambda A'\bfrho
         \,-\, (I-zA)^{-1}A'\bfrho}_{\X_{\Lambda'}}
   \,\to\, 0 \qquad \text{as $|\Lambda|\to\infty$}
\edm
holds true uniformly for all $v\in\V$ with $\norm{v}_\V\leq 1$.
\end{lemma}

Note that if one replaces $A'\bfrho$ by $\bfe$ in \req{term-I} then the
resulting assertion is that $\bfrho_\Lambda$ converges uniformly to $\bfrho$
in the given box $\Lambda'$ as $|\Lambda|\to\infty$, cf.~\req{KS-prime}.
In fact, one can reuse the corresponding proof of 
Ruelle~\cite[Theorem~4.2.3]{Ruel69} with $\bfalpha=A'\bfrho$ instead of 
$\bfalpha=\bfe$ throughout to verify \req{term-I}. 
The proof of Lemma~\ref{Lem:term-I} is then an easy consequence
because the rate of convergence only depends on the norm of $\bfalpha$, and
\be{Astrichrho-bd}
   \norm{A'\bfrho}_{\X_{\R^3}}
   \,\leq\, \norm{\partial A}_{\L(\V,\L(\X_{\R^3}))}
            \norm{v}_\V \norm{\bfrho}_{\X_{\R^3}}
\ee
is uniformly bounded for $\norm{v}_\V\leq 1$; take note that an upper bound
for \req{Astrichrho-bd} can be chosen in a way to be also an appropriate 
bound for $\norm{A_\Lambda'\bfrho_\Lambda}_{\X_\Lambda}$.

\begin{xproof}{\ref{Thm:B}}
We now turn to a proof of \req{term-II}. To this end we consider a finite
number of nested boxes
\be{nested-boxes}
   \Lambda' = \Lambda_{k_0} \subset \Lambda_{k_0-1} \subset\dots \subset 
   \Lambda_1 \subset \Lambda\,,
\ee
each of them centered at the origin, and $\Lambda_k$ having sides of length
\bdm
   \ell_k \,=\, \ell_0+2(k_0-k)d\,, \qquad k=1,\dots,k_0\,,
\edm
where $\ell_0$ is the length of the sides of the given box $\Lambda'$.
The particular number $k_0$ will be chosen later; see \req{term-II-est} below.
With each box $\Lambda_k$ we associate the Banach space 
$\X_k=\X_{\Lambda_k}$, the corresponding projector $P_k=P_{\Lambda_k}$ of \req{P},
and the operator $A_k=A_{\Lambda_k}\in\L(\X_k)$. As before we refer to
$A=A_{\R^3}$ for the operator corresponding to the full space.

At this stage we quote another auxiliary result, 
cf.~\cite[(2.41) in Chapter~4]{Ruel69}, namely that
\be{final-2}
\begin{aligned}
   \norm{P_{k+1}A_1^k P_1
         \,-\, P_{k+1}A_\Lambda^k}_{\L(\X_{\Lambda},\X_{k+1})}
   &\,\leq\, 2k e^{2\beta B+1} \norm{A_\Lambda}_{\L(\X_{\Lambda})}^{k-1} \cbetad \\
   &\,\leq\, 2k \cbeta^{k-1} e^{k(2\beta B+1)} \cbetad
\end{aligned}
\ee
for $k=0,1,\dots,k_0-1$.
In \cite{Ruel69} this result was proved with $A_\Lambda$ replaced by $A$
(and in $\L(\X_{\R^3},\X_{k+1})$), however, the argument given there
does not need any modification to establish \req{final-2} as it stands.
From \req{final-2} it follows that
\bdmal
   &\norm{A_\Lambda^kP_\Lambda
          (A'\bfrho-A_\Lambda'\bfrho_\Lambda)}_{\X_{k+1}}\\[1ex]
   &\ \,\leq\, \norm{(A_1^k P_1 - A_\Lambda^k)P_\Lambda
                 (A'\bfrho-A_\Lambda'\bfrho_\Lambda)}_{\X_{k+1}}
           \,+\, \norm{A_1^kP_1
                     (A'\bfrho-A_\Lambda'\bfrho_\Lambda)}_{\X_{k+1}}\\[1ex]
   &\ \,\leq\, 2k\cbeta^{k-1}e^{k(2\beta B+1)} \cbetad 
           \norm{A'\bfrho-A_\Lambda'\bfrho_\Lambda}_{\X_\Lambda}
           \,+\, \cbeta^ke^{k(2\beta B+1)} 
                 \norm{A'\bfrho-A_\Lambda'\bfrho_\Lambda}_{\X_1}\,,
\edmal
and hence, the representation
\bdmal
   &(I-zA_\Lambda)^{-1}P_\Lambda(A'\bfrho-A_\Lambda'\bfrho_\Lambda)
    \\[1ex]
   &\qquad
    \,=\, \sum_{k=0}^{k_0-1} z^k A_\Lambda^kP_\Lambda
                             (A'\bfrho-A_\Lambda'\bfrho_\Lambda)
          \,+\, \sum_{k=k_0}^\infty z^k A_\Lambda^kP_\Lambda
                                     (A'\bfrho-A_\Lambda'\bfrho_\Lambda)
\edmal
leads to the upper bound
\bdmal
   &\norm{(I-zA_\Lambda)^{-1}P_\Lambda
          (A'\bfrho-A_\Lambda'\bfrho_\Lambda)}_{\X_{\Lambda'}}
    \\[1ex]
   &\quad
    \,\leq\, \sum_{k=0}^{k_0-1} 
                z^k \norm{A_\Lambda^kP_\Lambda
                       (A'\bfrho-A_\Lambda'\bfrho_\Lambda)}_{\X_{k+1}}
    \,+\, \sum_{k=k_0}^\infty z^k \norm{A_\Lambda}_{\X_\Lambda}^k
                            \norm{A'\bfrho-A_\Lambda'\bfrho_\Lambda}_{\X_\Lambda}
    \\[1ex]
   &\quad
    \,\leq\, \sum_{k=0}^{k_0-1} 
                2k\cbeta^{k-1} (ze^{2\beta B+1})^k\cbetad\, 
                \norm{A'\bfrho-A_\Lambda'\bfrho_\Lambda}_{\X_\Lambda}\\
   &\quad
    \phantom{\,\leq\,}
    \ +\, \sum_{k=0}^{k_0-1} 
             (z\cbeta e^{2\beta B+1})^k
                \norm{A'\bfrho-A_\Lambda'\bfrho_\Lambda}_{\X_1}
    \,+\, \sum_{k=k_0}^\infty 
             (z\cbeta e^{2\beta B+1})^k
             \norm{A'\bfrho-A_\Lambda'\bfrho_\Lambda}_{\X_\Lambda}\,.
\edmal
By virtue of \req{Astrichrho-bd} $\norm{A'\bfrho}_{\X_{\R^3}}$ and 
$\norm{A_\Lambda'\bfrho_\Lambda}_{\X_\Lambda}$ are uniformly bounded, hence there
exists a constant $c>0$ (depending only on $z$ and $\beta$) with
\be{term-II-est}
\begin{aligned}
   &\norm{(I-zA_\Lambda)^{-1}P_\Lambda
          (A'\bfrho-A_\Lambda'\bfrho_\Lambda)}_{\X_{\Lambda'}}
    \\[1ex]
   & \qquad 
     \,\leq\, c\bigl(
              \cbetad \,+\, \norm{A'\bfrho-A_\Lambda'\bfrho_\Lambda}_{\X_1}
                      \,+\, (z\cbeta e^{2\beta B+1})^{k_0}\bigr)\,,
\end{aligned}
\ee
provided that $\norm{v}_\V\leq 1$ and $z$ and $\beta$ satisfy \req{z}.

Now let $\eps>0$ be given. By virtue of \req{z} and \req{cbetad} we can
choose $d$ and $k_0$ so large that the first and third summand on the 
right-hand side of \req{term-II-est} are both smaller than $\eps/(3c)$. 
A corresponding choice of nested boxes $\Lambda_k$ of \req{nested-boxes}
is possible provided that $\Lambda$ is sufficiently big.
In fact, making $|\Lambda|$ even larger 
the second term on the right-hand side of \req{term-II-est} will also 
become less than $\eps/(3c)$ by virtue of Lemma~\ref{Lem:Aprimelimit},
uniformly for $v\in\V$ with $\norm{v}_\V\leq 1$. 
Accordingly, there holds
\bdm
   \norm{(I-zA_\Lambda)^{-1}P_\Lambda
         (A'\bfrho-A_\Lambda'\bfrho_\Lambda)}_{\X_{\Lambda'}}
   \,\leq\,\eps
\edm
for $|\Lambda|$ sufficiently large, which yields \req{term-II}.
Again we emphasize that the size of $|\Lambda|$ to achieve a given bound
$\eps>0$ does not depend on $v$ as long as $\norm{v}_\V\leq 1$.

As mentioned before, a combination of \req{term-II} and \req{term-I},
cf.~Lemma~\ref{Lem:term-I}, implies that
\bdm
   \norm{(\partial\bfrho_\Lambda)v-(\partial\bfrho)v}_{\X_{\Lambda'}}
   \,\to\, 0 \qquad \text{as $|\Lambda|\to \infty$}\,,
\edm
uniformly for all $v\in\V$ with $\norm{v}_\V\leq 1$.
\end{xproof}

\section{Explicit computation of the derivatives of the singlet and
pair distribution functions}
\label{Sec:explicit}
In the sequel we are going to determine more suitable 
and implementable formulae for the derivatives of the first two molecular 
distribution functions in a finite size box $\Lambda\subset\R^3$.

For a given $v\in\V$ with $\norm{v}_{\V}\leq t_0$ we write
\be{VN}
   V_N(\RR_N) \,=\, \sum_{1\leq i<j\leq N} v(|R_i-R_j|)
\ee
in analogy to \req{pairpotential}. Then, to begin with,
a straightforward formal computation
provides the derivative $\partial\,\Xi_\Lambda$ of the grand canonical 
partition function~\req{Xi}:
\bdm
   (\partial\,\Xi_\Lambda)v \,=\, 
   -\beta \sum_{N=2}^\infty \frac{z^N}{N!}
             \int_{\Lambda^N} V_N(\RR_N)e^{-\beta U_N(\RR_N)}\dRR_N\,.
\edm
That $\partial\,\Xi_\Lambda$ is a Fr\'echet derivative in $\L(\V,\R)$ can
readily be checked by following the line of argument 
of the proof of Lemma~\ref{Lem:d-derivative}. 
Since $V_N$ is defined by a pairwise interaction of identical particles
we can rewrite this derivative in the simpler form
(cf., e.g., Ben-Naim~\cite[Sect.~3.1]{BenN06}, or the arguments 
employed below),
\bdm
   (\partial\,\Xi_\Lambda)v \,=\, 
   -\frac{\beta}{2} \int_{\Lambda}\int_\Lambda
       v(|R-R'|) \rho_\Lambda^{(2)}(R,R')\dR\!\dR'\,,
\edm
which is amenable to numerical computations.

The argument of Lemma~\ref{Lem:d-derivative} can also be used to determine
the Fr\'echet derivatives $\partial Z_\Lambda^{(m)}\in\L(\V,L^\infty(\Lambda^m))$
of the numerators
\bdm
   Z_\Lambda^{(m)}(\RR_m) 
   \,=\, \sum_{N=m}^\infty \frac{z^N}{(N-m)!}
         \int_{\Lambda^{N-m}}\!\! 
            e^{-\beta U_N(\RR_N)}
         \dRR_{m,N}
\edm
of the molecular distribution functions $\rho_\Lambda^{(m)}$, namely
\bdm
   \bigl((\partial Z_\Lambda^{(m)})v\bigr)(\RR_m)
   \,=\, - \beta\sum_{N=m}^\infty\frac{z^N}{(N-m)!}
               \int_{\Lambda^{N-m}} V_N(\RR_N)e^{-\beta U_N(\RR_N)}\dRR_{m,N}\,.
\edm
For $m=1$ we can utilize \req{VN} and the fact that individual particles are
indistinguishable to reformulate 
\bdmal
   &\bigl((\partial Z_\Lambda^{(1)})v\bigr)(R_1)
    \,=\, -\beta\sum_{N=2}^\infty \frac{z^N}{(N-1)!}
          \int_{\Lambda^{N-1}}\!\sum_{1\leq i<j\leq N}\!v(|R_i-R_j|)\,e^{-\beta U_N(\RR_N)}
          \dRR_{1,N}\\[1ex]
   &\qquad
    \,=\, -\beta\sum_{N=2}^\infty \frac{z^N}{(N-1)!} 
          \Biggl((N-1)\int_\Lambda v(|R_1-R_2|) 
                \int_{\Lambda^{N-2}}e^{-\beta U_N(\RR_N)}\dRR_{2,N}\dR_2\\
   &\qquad \qquad \qquad \qquad \qquad \quad \ \,
          +\! \sum_{2\leq i<j\leq N}
              \int_{\Lambda^{N-1}} v(|R_i-R_j|)
                                   e^{-\beta U_N(\RR_N)}\dRR_{1,N}
          \Biggr)\\[1ex]
   &\qquad
    \,=\, -\beta\sum_{N=2}^\infty \frac{z^N}{(N-2)!} 
          \int_\Lambda v(|R_1-R_2|) 
          \int_{\Lambda^{N-2}}e^{-\beta U_N(\RR_N)}\dRR_{2,N}\dR_2\\[1ex]
   &\qquad \phantom{\,=\,}\ 
          -\frac{\beta}{2}\sum_{N=3}^\infty \frac{z^N}{(N-3)!}
              \int_\Lambda\int_\Lambda v(|R_2-R_3|)
              \int_{\Lambda^{N-3}} e^{-\beta U_N(\RR_N)}\dRR_{3,N}\dR_3\!\dR_2\\[1ex]
   &\qquad
    \,=\, -\beta \int_\Lambda v(|R_1-R_2|)Z_\Lambda^{(2)}(\RR_2)\dR_2
          \,-\,\frac{\beta}{2}
           \int_\Lambda\int_\Lambda v(|R_2-R_3|)Z_\Lambda^{(3)}(\RR_3)\dR_3\!\dR_2\,.
\edmal
Likewise we obtain 
\bdmal
   &\bigl((\partial Z_\Lambda^{(2)})v\bigr)(R_1,R_2)
    \,=\, -\beta\sum_{N=2}^\infty\frac{z^N}{(N-2)!}
          \int_{\Lambda^{N-2}} V_N(\RR_N)e^{-\beta U_N(\RR_N)}\dRR_{2,N}\\[1ex]
   &\qquad
    \,=\, -\beta\sum_{N=2}^\infty\frac{z^N}{(N-2)!} 
          \Biggl(
             v(|R_1-R_2|)\int_{\Lambda^{N-2}} e^{-\beta U_N(\RR_N)}\dRR_{2,N}
                    \\[1ex]
   &\qquad \qquad 
          \,+\, (N-2)\int_\Lambda v(|R_1-R_3|) 
                \int_{\Lambda^{N-3}}e^{-\beta U_N(\RR_N)}\dRR_{3,N}\dR_3\\[1ex]
   &\qquad \qquad 
          \,+\, (N-2)\int_\Lambda v(|R_2-R_3|) 
                \int_{\Lambda^{N-3}}e^{-\beta U_N(\RR_N)}\dRR_{3,N}\dR_3\\[1ex]
   &\qquad \qquad 
          \,+\,
   \frac{(N-2)(N-3)}{2} \int_\Lambda\int_\Lambda v(|R_3-R_4|)
   \int_{\Lambda^{N-4}} e^{-\beta U_N(\RR_N)}\dRR_{4,N} \dR_4\!\dR_3
          \Biggr)\,,
\edmal
which can be rewritten as
\bdmal
   &\bigl((\partial Z_{\Lambda}^{(2)})v\bigr)(R_1,R_2)
    \,=\, -\beta\, v(|R_1-R_2|) Z_{\Lambda}^{(2)}(\RR_2)\\[1ex]
   &\qquad
    \,-\, \beta\int_\Lambda v(|R_1-R_3|)Z_{\Lambda}^{(3)}(\RR_3)\dR_3
    \,-\, \beta\int_\Lambda v(|R_2-R_3|)Z_{\Lambda}^{(3)}(\RR_3)\dR_3 \\[1ex]
   &\qquad
    \,-\, \frac{\beta}{2}
           \int_\Lambda\int_\Lambda 
              v(|R_3-R_4|) Z_{\Lambda}^{(4)}(\RR_4) \dR_4\!\dR_3\,.
\edmal

After these preparations we can employ the quotient rule to obtain
\be{drhodU1}
\begin{aligned}
   \bigl((\partial \rho_\Lambda^{(1)})v\bigr)(R_1)
   \,=\,& \frac{1}{\Xi_\Lambda}\bigl((\partial Z_\Lambda^{(1)})v\bigr)(R_1)
            \,-\, \rho_{\Lambda}^{(1)}(R_1)\,
                \frac{(\partial\,\Xi_\Lambda)v}{\Xi_\Lambda}\\[1ex]
   \,=\,& -\!\beta 
              \int_\Lambda v(|R_1-R'|)\rho_\Lambda^{(2)}(R_1,R')\dR' \\[1ex]
   &      - \frac{\beta}{2}
            \int_\Lambda\int_\Lambda 
               v(|R_1'-R_2'|) \rho_\Lambda^{(3)}(R_1,R_1',R_2')
            \dR_2'\!\dR_1'\\[1ex]
   &      + \frac{\beta}{2}
            \int_\Lambda\int_\Lambda 
               v(|R_1'-R_2'|) \rho_\Lambda^{(1)}(R_1)\rho_\Lambda^{(2)}(R_1',R_2')
            \dR_2'\!\dR_1'
\end{aligned}
\ee
and
\begin{subequations}
\label{eq:drhodU2}
\begin{align}
\nonumber
   \bigl((\partial \rho_\Lambda^{(2)})v\bigr)&(R_1,R_2)
   \,=\, \frac{1}{\Xi_\Lambda}\bigl((\partial Z_\Lambda^{(2)})v\bigr)(R_1,R_2)
          \,-\, \rho_\Lambda^{(2)}(R_1,R_2)\,
                \frac{(\partial\,\Xi_\Lambda)v}{\Xi_\Lambda} \\[1ex]
\nonumber
   \,=\,& -\!\beta\, v(|R_1-R_2|) \rho_\Lambda^{(2)}(R_1,R_2)\\[1ex]
\label{eq:drhodUa}
        &- \beta
           \int_\Lambda v(|R_1-R'|)\rho_\Lambda^{(3)}(R_1,R_2,R')\dR' \\[1ex]
\label{eq:drhodUb}
        &- \beta
           \int_\Lambda v(|R_2-R'|)\rho_\Lambda^{(3)}(R_1,R_2,R')\dR' \\[1ex]
\label{eq:drhodUc}
        &- \beta
           \int_\Lambda\int_\Lambda 
              \frac{1}{2}\,v(|R_1'-R_2'|) \rho_\Lambda^{(4)}(R_1,R_2,R_1',R_2')
           \dR_1'\!\dR_2'\\[1ex]
\nonumber
        &+ \frac{\beta}{2}
           \int_\Lambda\int_\Lambda 
              v(|R_1'-R_2'|) \rho_\Lambda^{(2)}(R_1,R_2)
                             \rho_\Lambda^{(2)}(R_1',R_2')
           \dR_2'\!\dR_1'\,.
\end{align}
\end{subequations}

\begin{remark}
\label{Rem:Lyubartsev}
\rm
To illustrate \req{drhodU2} imagine the situation of a fixed pair of 
coordinates $R_1,R_2\in\Lambda$ with $R_1\neq R_2$, when $v=v(r)$ is a delta
distribution located in $r'>0$. Then the three integrals in
\req{drhodUa}--\req{drhodUc} provide the expected number of events -- up to
a factor $4\pi r'{}^2$ due to the polar coordinate transformation -- 
of encountering at the same time a pair of particles at $R_1$ and $R_2$ and
a different pair of particles with distance $r'$. The three individual terms
account for events where
\begin{itemize}
\item[\req{drhodUa}:]
  $R_1$ is one of the two members of the second pair of particles,
\item[\req{drhodUb}:]
  $R_2$ is one of the two members of the second pair of particles,
\item[\req{drhodUc}:]
  all four particles involved are different;
  during integration every second pair of particles is counted twice,
  hence the extra factor $1/2$.
\end{itemize}
This agrees with the formula provided in \cite{LyLa95}.
\fin
\end{remark}

Finally we mention that neither of the two representations of
$(\partial\rho_\Lambda^{(1)})v$ and $(\partial\rho_\Lambda^{(2)})v$
has a straightforward extension to the thermodynamical limit.
To extend these formulae to the thermodynamical limit the double integrals
need to be recombined. This will be reconsidered in part II 
of this work~\cite{Hank16b}.

\section{Conclusions}
\label{Sec:finale}
We have shown that for a grand canonical ensemble of idential point-like 
particles in thermodynamical equilibrium the molecular distribution 
functions $\rho_\Lambda^{(m)}$ and their thermodynamical limits $\rho^{(m)}$ 
are differentiable with respect to the underlying pair potential 
and the $L^\infty$ norm of the distribution functions. To do so we have 
considered pair potentials that satisfy Assumption~\ref{Ass:u} which is 
slightly stronger than just being stable and regular; for these potentials 
we can treat the entire regime of small activities $z>0$ for which the 
thermodynamical limit of the molecular distribution functions is known to 
be well-defined. In physical terms these activities correspond to a region
without phase transitions, i.e., the \emph{gas phase} of this molecular fluid, 
cf.~\cite{Ruel69}. 

Assumption~\ref{Ass:u} comes with a natural topology to study perturbations 
$v$ of $u$: On the one hand perturbations must not be too strong to 
violate the repulsive nature of the potential when particles get close,
hence, $|v|$ must be bounded by $u$ near the origin; on the other hand,
perturbations must be sufficiently small for distant particles
to maintain regularity. 
We mention, though, that 
the
choice of $\uo$ allows rather general decay rates of the underlying potential 
$u$ and its perturbations $v$; as a consequence the latter can by far dominate 
the underlying potential near infinity.

We have also proved that the derivative of $\rho_\Lambda^{(m)}$ converges 
compactly to the derivative of $\rho^{(m)}$ as $|\Lambda|\to\infty$.
This justifies the approximate computation of, say, the derivative of the 
radial distribution function (an often-used structural quantity in 
chemical physics) by numerical particle simulations in a finite size box,
as suggested in \cite{LyLa95}. 

In a subsequent paper~\cite{Hank16b} we will reconsider
the pair distribution function $\rho_\Lambda^{(2)}$ and the corresponding radial 
distribution function $g=g(r)$ because they call for the investigation of
differentiability in another function space that reflects the property
that $g(r)\to 1$ as $r\to\infty$ for potentials $u$ that satisfy 
Assumption~\ref{Ass:u}.
The methods that come to use in \cite{Hank16b} employ cluster expansions
and are of completely different nature than the ones that 
have been utilized here.


\end{document}